\documentclass[aps,pra,letterpaper,superscriptaddress,twocolumn]{revtex4}
\usepackage[utf8]{inputenc}
\usepackage{amsmath}
\usepackage{graphicx}
\usepackage{color}
\usepackage{amsfonts}
\usepackage{amssymb}
\usepackage{float}
\usepackage{amsthm}
\usepackage{makeidx}
\usepackage{graphicx}
\usepackage{comment}
\usepackage[normalem]{ulem}
\newtheorem{theorem}{Theorem}
\newtheorem{corollary}{Corollary}
\newtheorem{lemma}{Lemma}
\newtheorem{definition}{Definition}

\newcommand{\CS}{|\text{CS}\rangle}
\newcommand{\INS}{| I \rangle}

\input{Qcircuit}
\begin{document}
\title{Universality of quantum computation with cluster states and (X,Y)-plane measurements}
\author{Atul Mantri}
\affiliation{Singapore University of Technology and Design, 8 Somapah Road, Singapore 487372}
\author{Tommaso F. Demarie}
\affiliation{Singapore University of Technology and Design, 8 Somapah Road, Singapore 487372}
\author{Joseph F. Fitzsimons}
\affiliation{Singapore University of Technology and Design, 8 Somapah Road, Singapore 487372}
\affiliation{Centre for Quantum Technologies, National University of Singapore, Block S15, 3 Science Drive 2, Singapore 117543}
\begin{abstract}
Measurement-based quantum computing (MBQC) is a model of quantum computation where quantum information is coherently processed by means of projective measurements on highly entangled states. Following the introduction of MBQC, cluster states have been studied extensively both from the theoretical and experimental point of view. Indeed, the study of MBQC was catalysed by the realisation that cluster states are universal for MBQC with (X,Y)-plane and Z measurements. Here we examine the question of whether the requirement for Z measurements can be dropped, while maintaining universality. We answer this question in the affirmative by showing that universality is possible in this scenario.
\end{abstract}
\maketitle

\section{Introduction}
Cluster states~\cite{Nielsen2006} are highly entangled quantum states that play the role of central resources in measurement-based quantum computing (MBQC). While the canonical understanding of quantum computation relies on the imagery and concepts of quantum circuits~\cite{Nielsen2000}, MBQC recreates the full toolbox of wires and gates by means of local adaptive quantum measurements on said cluster states. This point of view is particularly appealing because it replaces the issue of coherently controlling quantum states with the less experimentally challenging action of creating an entangled resource at the initial stage of the computation. Indeed cluster states can be created efficiently in any system with a quantum Ising-type interaction (at very low temperatures) between two-state particles in a lattice configuration. In its original formulation~\cite{Raussendorf2001}, universality of MBQC was shown using a cluster state and single-qubit projective measurements: Precisely, this universality proof makes use of $(X,Y)$-plane measurements as well as $Z$-basis (computational) measurements used to remove redundant qubits from the cluster state. Such formulation of the MBQC universality proof is still considered today as \emph{the benchmark proof}, particularly because of its simple intuitive power. Nonetheless it seems fair to ask whether universality can be achieved limiting the measurements to a single plane of the Bloch sphere.

Beyond more fundamental reasons, the motivation to reduce the angle set to a single plane follows from the prevalent formulation of delegated quantum computing (DQC) protocols. Critically, the adaptive nature of MBQC proved to be of great importance in the development of secure DQC: The first universal and unconditionally secure blind quantum computing protocol presented in~\cite{Broadbent2009} is entirely based on MBQC and exploits the  interaction between a client and a server to protect the client's information. In this protocol, measurements are performed by the server and belong solely to the $(X,Y)$-plane. This succeeds because the resource state considered is a brickwork state, which can be prepared by performing an appropriate pattern of $Z$-measurements on a cluster state. More recently, there has been a plethora of work on blind DQC grounded on MBQC, see for example~\cite{Barz2012,Barz2013,Fitzsimons2013,Morimae2013MO,Mantri2016} and references therein (for a more general discussion about DQC we refer the readers to~\cite{Dunjko2014}).

The idea of secure DQC is motivated by very practical issues: One could safely anticipate that when quantum computers will be available they will be hosted by large institutions offering their services in a cloud fashion~\cite{IBM2016}. Blind protocols are there to allow for a client with limited quantum technologies to access the full-power of quantum computers while protecting the privacy of her information. In this letter we show that a 2-dimensional cluster state is universal for quantum computation with measurements restricted to the $(X,Y)$-plane. This result is general, novel to the best of our knowledge, and implies that every blind DQC protocol rooted on MBQC can use a cluster state as resource, with no fundamental necessity to introduce more particular states such as the brickwork state. 

Here we try to keep the formalism needed for our proof at a minimum. All the relevant concepts are introduced, but a certain level of familiarity with the ideas of measurement-based quantum computing is assumed.

\section{Definitions and Notations}
We start by defining the notation used throughout the text. The identity gate is $\hat{\mathbb{I}}$, we use the symbols $\hat{X}, \hat{Z}$ for the Pauli gates, $\hat{H}$ for the Hadamard gate and $\hat{R}_Z(\theta) = \text{exp}(-i \frac{\theta}{2} \hat{Z})$ for a generic $Z$-rotation by an angle $\theta$. The states $|0\rangle$ and $|1\rangle$ form the computational basis and they are eigenstates of $\hat{Z}$. The eigenstates of $\hat{X}$ are $| +\rangle$ and $|-\rangle$, with $| +\rangle = \hat{H} |0\rangle$ and $| -\rangle = \hat{H} |1\rangle$. One of the entangling gates for the $|\pm\rangle$ basis is the two-qubit controlled-$\hat{Z}$ operator, given by $\text{Ctrl-Z} = |0\rangle \langle0| \otimes \hat{\mathbb{I}} + |1\rangle \langle1| \otimes \hat{Z}$. We also define a two-qubit gate with weight $\alpha$ as $\hat{R}_{ZX}(\alpha) = \exp(- i \frac{\alpha}{2} \hat{Z} \otimes \hat{X})$. Note that if a $Z$-rotation (with angle $\theta$) is applied to a $|+\rangle$ state (or to its orthogonal state $|-\rangle$) one gets, \[\hat{R}_Z(\theta) |\pm\rangle = |\pm_{\theta}\rangle := \frac{1}{\sqrt{2}}(|0\rangle \pm \exp(i \theta)|1\rangle).\] Analogously, a single-qubit projective $(X,Y)$-plane measurement with angle $\theta$ is equivalent to a measurement in the basis $ \{ |+_{\theta}\rangle, |-_{\theta}\rangle\}$. We now give two definitions for cluster states:

\begin{definition}(Cluster State)
A cluster state $|\text{CS}_{n \times m}\rangle$, is an entangled state of $n \times m$ qubits constructed as follows:
\begin{enumerate}
\item Prepare all the qubits in the state $|+\rangle$ and assign to each qubit a unique index (i,j), i being a row ($i \in [n]$) and j being a column ($j \in [m]$).
\item For each row ($1 \leq i \leq n$),  apply the operator $\text{Ctrl-Z}$ on qubits (i,j) and (i,j+1) where $1 \leq j \leq m-1$.  
\item For each column ($1 \leq j \leq m$), apply the operator $\text{Ctrl-Z}$ on qubits (i,j) and (i+1,j) where $1 \leq i \leq n-1$.  
\end{enumerate}
\end{definition}

\begin{definition}(Open-ended Cluster State)
A cluster state $|\text{OCS}_{n \times m}\rangle$, is an entangled state of $n \times m$ qubits constructed as follows :
\begin{enumerate}
\item Prepare all the qubits in the state $|+\rangle$ and assign to each qubit a unique index (i,j), i being a row ($i \in [n]$) and j being a column ($j \in [m]$).
\item For each row ($1 \leq i \leq n$),  apply the operator $\text{Ctrl-Z}$ on qubits (i,j) and (i,j+1) where $1 \leq j \leq m-1$.  
\item For each column ($1 \leq j \leq m-1$), apply the operator $\text{Ctrl-Z}$ on qubits (i,j) and (i+1,j) where $1 \leq i \leq n-1$.
\end{enumerate}
\end{definition}

\section{MBQC - Preliminaries}
The workhorse of MBQC is the well-known concept of one-bit teleportation~\cite{gottesman1999quantum,Gottesman1999, Zhou2000}, which is also particularly helpful to visualise our universality proof. We will present this idea using two qubits initialised in the following state: $ \text{Ctrl-Z} (|\psi\rangle \otimes |+\rangle)$. Since $\text{Ctrl-Z}$ is a symmetric operator any of the two qubits can be chosen as control (or target) qubit. This is conceptually equivalent to preparing a cluster state in MBQC where the first qubit has been replaced by a generic state $|\psi\rangle$. The teleportation happens when the first qubit ($|\psi\rangle$) is measured in the $\{ |+\rangle, |-\rangle \}$ basis: Circuit-wise, this is represented by applying first an Hadamard gate and then measuring in the computational basis as shown in Figure \ref{fig:1bit}. Importantly, after the measurement the state of the second qubit is equal to $\hat{X}^m \hat{H} |\psi\rangle$, where $\hat{X}^m$ is a Pauli correction induced by the measurement outcome $m$.

\begin{figure}
\[
\Qcircuit @C=1em @R=1em{
\lstick{\ket{\psi}} & \ctrl{2} & \gate{H} & \meter & \rstick{m} \cw \\
&  & & & & & & &  \\
\lstick{\ket{+}}  & \control \qw & \qw & \qw & & \hat{X}^{m}\hat{H}\ket{\psi} \\
}
\]
\caption{The quantum circuit for the one-bit teleportation scheme.}
\label{fig:1bit}
\end{figure}
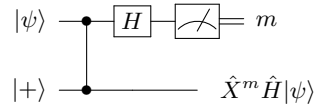

A generalised one-bit teleportation circuit corresponds to performing a measurement in the $(X,Y)$-plane instead of the $\{|+\rangle, |-\rangle\}$ basis measurement. This circuit is shown in Figure \ref{fig:1bitg}. On the left hand side is shown the circuit corresponding to a measurement of angle $\theta$ in the $(X,Y)$-plane on the first qubit of the two-qubit state.  
Because any generic $\hat{R}_Z(\theta)$ rotation commutes with the $\text{Ctrl-Z}$ gate, one can easily transform the circuit on the left to an instance of the one-bit teleportation circuit by simply updating the initial state. Note that the resulting state at the end of the circuit inherits the $Z$-rotation introduced by the measurement.

\begin{figure}
\[
\Qcircuit @C=1em @R=1em{
\lstick{\ket{\psi}} & \ctrl{2} & \gate{HR_{Z}(\theta)} & \meter & \rstick{m} \cw \\
&  & & & & & & &  & \push{\rule{.3em}{0em}\equiv\rule{.3em}{0em}}  \\
\lstick{\ket{+}}  & \control \qw & \qw & \qw & & & \hat{X}^{m}\hat{H}\hat{R}_{Z}(\theta)\ket{\psi}\\
&  & & & & & & &  &  \\
\lstick{\hat{R}_{Z}(\theta)\ket{\psi}} & \ctrl{2} & \gate{H} & \meter & \rstick{m} \cw \\
&  & & & & & & &  &  \\
\lstick{\ket{+}}  & \control \qw & \qw  &\qw & & & \hat{X}^{m}\hat{H}\hat{R}_{Z}(\theta)\ket{\psi}\\
}
\]
\caption{Generalised one-bit teleportation circuit. LHS of the equivalence sign: measuring the first qubit with the $\theta$ angle in the (X,Y) plane, and RHS of the equivalence sign: applying a $Z$-rotation on the first qubit and measuring it in the X-basis. The two circuits are equivalent.}
\label{fig:1bitg}
\end{figure}
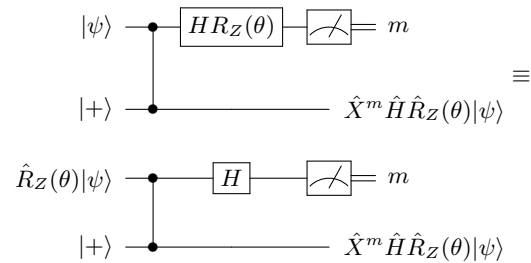

The customary understanding of MBQC is formulated in terms of a measurement \emph{pattern}: One defines an input and output set on the resource state such that the measurements transform the corresponding input state into the desired outcome, identified by the qubits of the output set. Intuitively, the quantum information is transformed by the same process that governs the generalised one-bit teleportation scheme. For the cluster state (open-ended or not), if the cardinality of the input set is $|I|$, then the input state of the computation corresponds to $|+\rangle^{\otimes |I|}$.  Note that this does not have to be the case in general, as shown for example in Figure~\ref{fig:1bit} and Figure~\ref{fig:1bitg}. Resource states can be concatenated, each representing the subroutine of a larger computation. Then, as dictated by the execution order, the output state of a previous resource state would correspond to the input state of the following one. Importantly, changing the input state of the MBQC resource state by replacing the $|+\rangle^{\otimes |I|}$ with a generic $|I|$-qubit state $| I \rangle$ does not affect the transformation induced by the measurements~\cite{Danos2007}. To define this concept properly, we give two modified definitions for cluster states:
\begin{definition}(Cluster State with generic input state)
A quantum state $|\text{inCS}_{n \times m}\rangle$, is a cluster state of $n \times m$ qubits with a well-defined input set of cardinality $|I|$. Additionally, the $|+\rangle^{\otimes |I|}$ qubits identified by the location of the input set are replaced by a generic $| I \rangle$ state before the entangling operations are performed. 
\end{definition}
\begin{definition}(Open-ended Cluster State with generic input state)
A quantum state $|\text{inOCS}_{n \times m}\rangle$, is an open ended cluster state of $n \times m$ qubits with a well-defined input set of cardinality $|I|$. Additionally, the $|+\rangle^{\otimes |I|}$ qubits identified by the location of the input set are replaced by a generic $| I \rangle$ state before the entangling operations are performed. 
\end{definition}

It is important to note that in MBQC measurements are in general adaptive in nature because of their inherent randomness, therefore future measurements might depend on previous outcomes. This is not problematic since suitable local corrections can be applied to move the computation back to the positive branch, i.e. where all the measurement outcomes are 0. We will therefore ignore the measurement dependencies in the rest of the manuscript and only deal with the unitary operations given by the positive branch of the measurement pattern~\cite{Danos2006,Danos2007}.

\section{Universality of cluster states with (X,Y)-plane measurements }
In this section we present a proof of universality for a cluster state with measurements constrained to the $(X,Y)$-plane. We start by proving a number of lemmas that will make the exposition of the main theorem easier. We start by introducing the notion of a universal gate set~\cite{Nielsen2000,boykin2000new}. The aim of the later proof is to show that it is possible to reproduce such a gate set by appropriate $(X,Y)$-plane measurements. 
\begin{lemma}
\label{lem:gateset}
Consider the gate set given by $\{ \exp(- i \frac{\theta}{2} \hat{Z}_{i} ), \exp(- i \frac{\theta}{2} \hat{X}_{i}), \exp(- i \frac{\theta}{2} \hat{Z}_{i} \otimes \hat{X}_{i+1})  \}$, with $i$ and $i+1$ adjacent qubits: It forms a universal set of quantum gates for quantum computing. 
\end{lemma}
\begin{proof}
The gates $\hat{R}_{Z}(\theta), \hat{R}_{X}(\theta)$ generate SU(2), and hence any single qubit operation can be implemented via a sequence of such gates. Together with the $\hat{R}_{ZX}(\theta)$ gate, these suffice to implement a CNOT gate between nearest neighbours. As a pair of qubits can be swapped through a sequence of three CNOT gates, the logical qubits can be permuted arbitrarily, allowing for CNOT gates to be implemented between arbitrary pairs of qubits. Thus the gate set considered here is equivalent to local unitaries together with CNOT gates, which has long been known to be universal~\cite{barenco1995elementary}.
\end{proof}

In the rest of the paper, we consider a $n \times m$ open-ended cluster state with generic input state $|\text{inOCS}_{n\times m}\rangle$. There can be two cases either $m > n$ or $m \leq n$. We only consider the former case with measurements executed column-by-column from left to right. We fix the first column to be the input state $\INS$ and the last ($m$-th) column to be the output state. All the $m-1$ columns preceding the output set are measured, and their qubits called operational qubits. An example for clarity is shown in Figure~\ref{fig:inCS}.

\begin{figure}
\includegraphics[width=150pt]{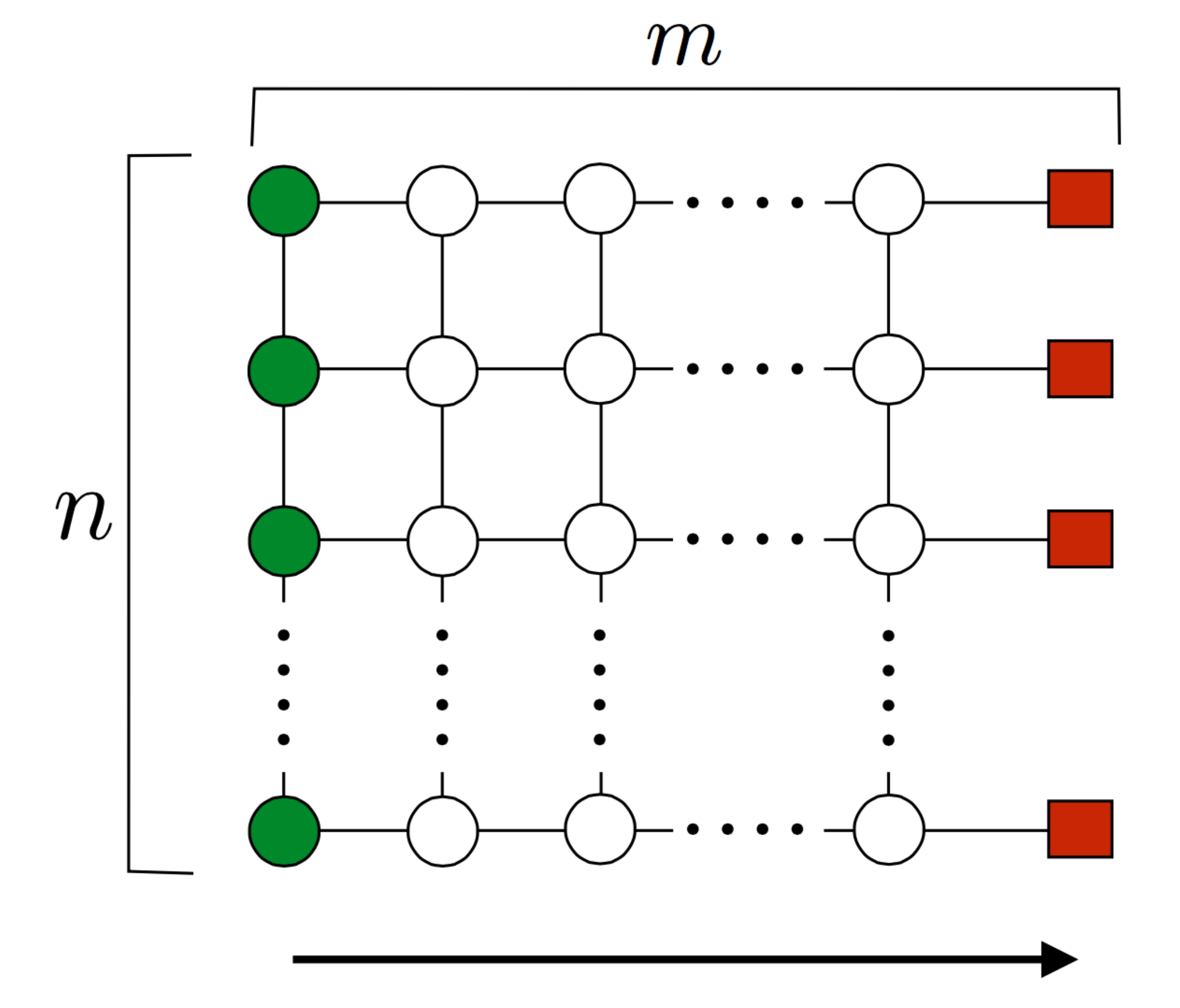}
\caption{A generic resource state $|\text{inOCS}_{n \times m}\rangle$. In green is shown the input set, in red the output set, while the arrow below indicates the sequential order of the measurements.}
\label{fig:inCS}
\end{figure}

\begin{lemma}
\label{lem:CS1+xmeas}
A resource state $|\text{inOCS}_{n\times m}\rangle$ together with $(X,Y)$-plane measurements $\{|+_\theta \rangle \langle +_\theta|\}$ along the $X$-axis, i.e. $\theta =0$, can be used to perform the unitary $\hat{U}$ on any input state $|I\rangle$, where $U$ is given by:
\begin{align}
\hat{U} = \prod_{j=1}^{m-1} \left( \otimes_{i=1}^{n}\hat{H}_{i} \right)_{j} \left( \otimes_{i=1}^{n-1} \text{Ctrl-Z}_{i,i+1} \right)_{j} \, .
\end{align}
\end{lemma}

\begin{proof}
From the one-bit teleportation scheme we can see that by measuring in the $\{ |+\rangle, |-\rangle \}$ basis all the operational qubits, the output of the computation is equal to $\hat{U}$, where $\hat{U}$ corresponds to the unitary implemented by the positive branch of the MBQC pattern. 
\end{proof}
 
Let us now define:
\begin{equation}
\hat{C}_n :=  \left( \otimes_{i=1}^{_{n}}\hat{H}_{i} \right)  \left( \otimes_{i=1}^{_{n-1}} \text{Ctrl-Z}_{i,i+1} \right) \, .
\end{equation}
Note that for a cluster state $|\text{inOCS}_{n \times m}\rangle$ the unitary operation $\hat{U}$ consists of $(m-1)$ repetitions of $\hat{C}_n $ applied on the input state $|I\rangle$. A simple example of such MBQC pattern is shown in Figure~\ref{fig:cluster} for the $2\times4$ cluster state $|\text{OCS}_{2 \times 4}\rangle$ with standard input state $|I\rangle = |++\rangle$.

To show that it is possible to reproduce the universal gate set from Lemma~\ref{lem:gateset}, we study the effects of performing (X,Y)-plane measurements with generic angle $\theta$ in different positions of the $|\text{inOCS}\rangle_{n\times (n+2)}$, for $m = n+2$. This choice of $m$ is justified by reasons of symmetry and does not affect the generality of the proof, which can be rewritten for generic values of $m$ at the cost of a less clear interpretation of the findings.

\begin{lemma}
\label{lem:zrot}
Consider a $n \times (n+2)$ open-ended cluster state $|\text{inOCS}_{n\times(n+2)}\rangle$, then the following statements are true:
\begin{enumerate}
\item When the $i$-th qubit of the first column is measured with a generic angle $\theta$ on the (X,Y)-plane, and all the other operational qubits are measured along the $X$-basis, the MBQC pattern implements a $Z$-rotation, $\hat{R}_{Z}(\theta)$, on the $(n+1-i)$-th qubit of the input state $|I\rangle$.
\item When the $i$-th qubit of the $(n+1)$-th column is measured with a generic angle $\theta$ on the (X,Y)-plane, and all the other operational qubits are measured along the $X$-basis, the MBQC pattern implements a $X$-rotation, $\hat{R}_{X}(\theta)$, on the  $(n+1-i)$-th qubit of the input state $|I\rangle$.
\item When the $i$-th qubit of the $p$-th column with $i=1$ and $1 < p < n+1$ is measured with a generic angle $\theta$ on the (X,Y)-plane, and all the other operational qubits are measured along the $X$-basis, the MBQC pattern implements an entangling gate, $\hat{R}_{ZX}(\theta)$, on the $n-p+1$ and $n-p+2$ qubits of the input state $|I\rangle$ for $i=1$. Analogously, for $i=n$ the same MBQC pattern implements the entangling gate $\hat{R}_{ZX}(\theta)$ on the $p$ and $p-1$ input qubits.
\end{enumerate}
\end{lemma}
\begin{proof}
To prove the lemma we look at how the corresponding quantum circuits change with the position on the state $|\text{inOCS}_{n\times(n+2)}\rangle$ of a single $(i,j)$-th qubit measured with an angle $\theta$. As stated above, all the other operational qubits are measured along the $X$-basis. This change of measurement basis can be equivalently written as a $Z$-rotation with angle $\theta$, i.e. $\hat{R}_{Z}(\theta) = \exp(-i \frac{\theta}{2} \hat{Z})$ on the $(i,j)$-qubit before the measurement. 

Using Lemma \ref{lem:CS1+xmeas}, we know that a measurement of all the qubits of one layer of $|\text{inOCS}_{n\times(n+2)}\rangle$ in the $X$-basis implements the $\hat{C}_n$ operator. Most importantly, since the $\hat{C}_n $ operator belongs to the Clifford group~\cite{Nielsen2000} it is easy to study how a single $Z$-rotation propagates through the circuit by using the canonical commutation relations of the Pauli matrices.

We define some helpful properties of the $\hat{C}_n$ operator. Inspired by models of quantum computation that exhibit particular mirror symmetries~\cite{Raussendorf2005,Fitzsimons2006}, we note that a repetition of $(n+1)$ $\hat{C}_n$ operators acts as a global mirror operation on the initial $n$-qubit state. Precisely, the following relations hold : $\hat{C}_n^{n+1}\hat{Z}_{i}|I\rangle = \hat{Z}_{\bar{i}}|I\rangle$ where the mirror qubit $\bar{i} = n+1-i$. Similarly, $\hat{C}_n^{n+1}\hat{X}_{i}|I\rangle = \hat{X}_{\bar{i}}|I\rangle$. More generally, the repetition of the $\hat{C}_n $ gate is used as a generalised swap gate as shown in~\cite{Fitzsimons2006,Raussendorf2005}. 

As always, examples are helpful to support the mathematical intuition: In Figure~\ref{fig:cluster}, we show the effects of a repetition of $\hat{C}_2$ gates on the familiar $|\text{OCS}_{2 \times 4} \rangle$ state. 

Using the mirror symmetry relations shown above, it is easy to see that a cluster state $|\text{inOCS}_{n \times (n+2)}\rangle$ whose $(i,1)$-th qubit is measured in the $(X,Y)$-plane with some angle $\theta$, while the remaining qubits are measured along the $X$-axis, implements a $\hat{R}_{Z_{n+1-i}}(\theta)|I\rangle$. Equally, measuring the $(i,n+1)$-th qubit in the $(X,Y)$-plane with some angle $\theta$ implements a $\hat{R}_{X_{n+1-i}}(\theta)|I\rangle$. This proves the first two statements of the lemma. 

\begin{figure*}
\includegraphics[width=\textwidth]{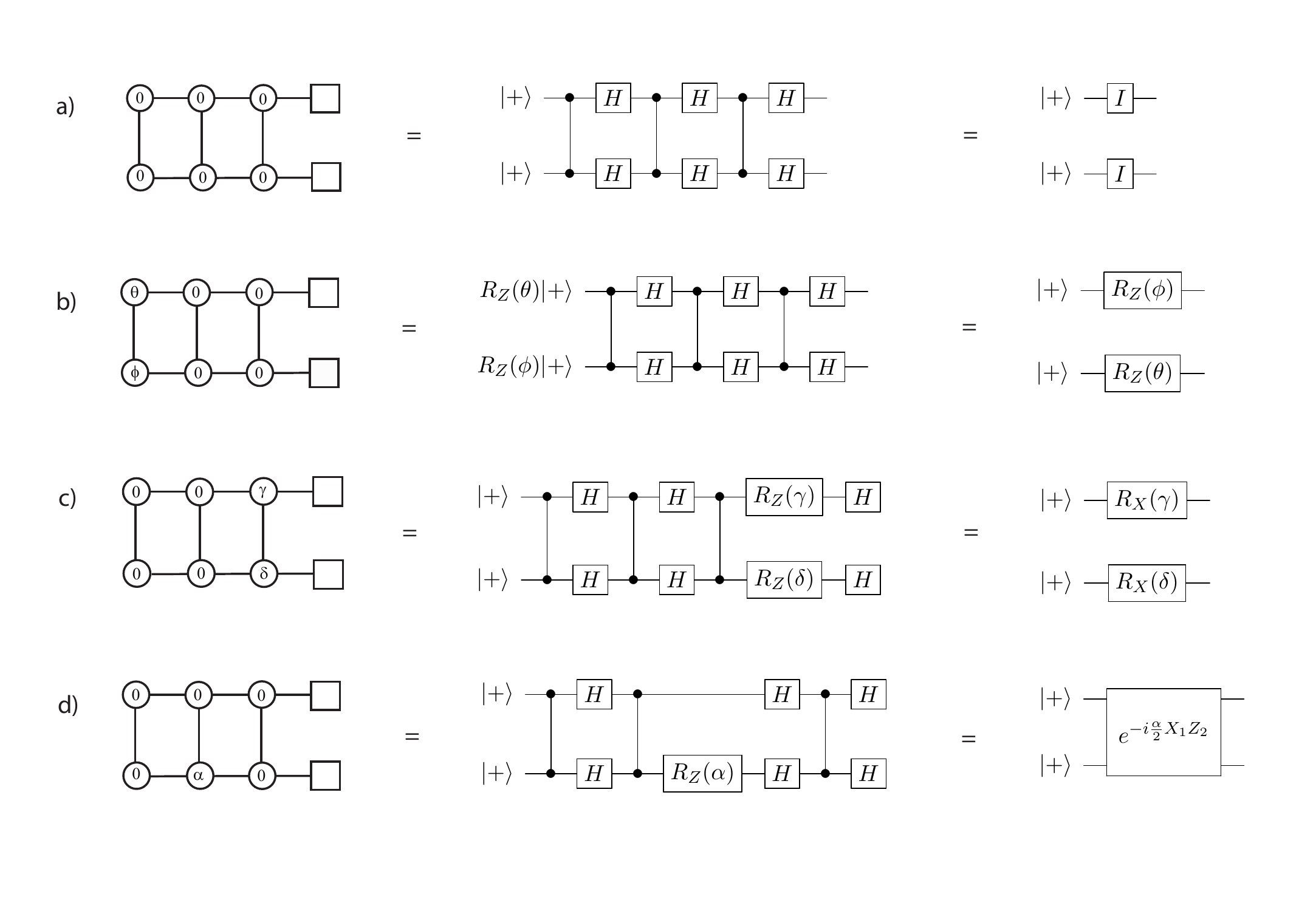}
\caption{Implementation of a) Identity gate, b) single-qubit Z rotation, c) single-qubit X rotation, and d) nearest neighbour entangling gate using $|\text{OCS}_{2 \times 4}\rangle$ and (X,Y)-plane measurements.  In these figures we use the convention that circles represent operational qubits, i.e. qubits that are measured during the computation, and squares represent output qubits, i.e. qubits left unmeasured at the end of the computation.}
\label{fig:cluster}
\end{figure*}

To prove the final statement, we analyse the situation when a rotation $\hat{Z}_{i}$ (or generally a $\hat{R}_{Z_{i}}$) is applied in the middle of a sequence of $\hat{C}_n$ operators. Explicitly, interposing an $(X,Y)$-measurement by an angle $\theta$ on the $(i,p)$-th qubit of a state $|\text{inOCS}_{n \times (n+2)}\rangle$ entirely measured along the $X$-basis implements $\hat{C}_n^{n-p+2}\hat{R}_{Z_{i}}(\theta)\hat{C}_n^{p-1}|I\rangle$. For a 2-qubit entangling gate it suffices to consider $i=1,n$.  We have already considered the case when $p = 1$ and $p = n+1$, and to derive a relation for any other $p$, we will use the following commutation relations and circuit identities (to simplify the notation we are not writing the identity gates): 

\begin{enumerate}
\item $\text{Ctrl-Z}_{(i,j)}\hat{Z}_{i} = \hat{Z}_{i}\text{Ctrl-Z}_{(i,j)}$ 
\item $\text{Ctrl-Z}_{(i,j)}\hat{X}_{i} = \hat{X}_{i}\hat{Z}_{j}\text{Ctrl-Z}_{(i,j)}$ 
\item $\text{Ctrl-Z}_{(i,j)}\hat{R}_{Z_{i}}(\theta) = \hat{R}_{Z_{i}}(\theta)\text{Ctrl-Z}_{(i,j)}$ 
\item $\hat{H}\hat{Z}\hat{H} = \hat{X}$
\end{enumerate}
Using the above circuit identities it is easy to see that:
\begin{align}
\hat{C}_n( \hat{Z}_{1} \otimes_{i=2}^{n} \hat{\mathbb{I}}_i) &=  ( \hat{X}_{1} \otimes \hat{Z}_{2} \otimes_{i=3}^{n} \hat{\mathbb{I}}_i)\hat{C}_n \, , \nonumber \\
\hat{C}_n (\otimes_{i=1}^{n-1} \hat{\mathbb{I}}_i \otimes \hat{Z}_{n}) &=  (\otimes_{i=1}^{n-2} \hat{\mathbb{I}}_i \otimes \hat{Z}_{n-1} \otimes \hat{X}_{n} )\hat{C}_n \, .
\end{align}

This relation can be extended to the case of a generic value of $p:1<p<n+1$ via a recursive application of the gate $\hat{C}_n$. Therefore, for $p' = n-p+2$ we have that
\begin{align}
&\hat{C}_n^{p'}( \hat{Z}_{1} \otimes_{i=2}^{n} \hat{\mathbb{I}}_i) =  ( \otimes_{i=1}^{n-p} \hat{\mathbb{I}}_i \otimes \hat{Z}_{n-p+1} \otimes \hat{X}_{n-p+2} \otimes_{i=n-p+3}^{n} \hat{\mathbb{I}}_i)\hat{C}_n^{p'} \, , \nonumber \\
&\hat{C}_n^{p'} (\otimes_{i=1}^{n-1} \hat{\mathbb{I}}_i \otimes \hat{Z}_{n}) = (\otimes_{i=1}^{p-2} \hat{\mathbb{I}}_i \otimes \hat{X}_{p-1} \otimes \hat{Z}_{p}  \otimes_{i=p+1}^{n} \hat{\mathbb{I}}_i )\hat{C}^{p'}_n \, .
\end{align}
Note that for $i \neq 1,n$ the resulting gate on the output state will still be an entangling gate, but it will have a more complicated form than the simple two-qubit gate one presented above. Since we are only interested in the universality proof we need not to discuss the most general case.

Using the procedure above one can steer a $\hat{Z}_{k}\otimes \hat{X}_{k+1}$ or $\hat{X}_{k} \otimes \hat{Z}_{k+1}$, for any $ 1\le k < n$, to an arbitrary position of the output state depending on the number $p$ of $\hat{C}_n$ operations in the circuit. The above results can be generalised by replacing the Pauli-$Z$ with any $Z$-rotation, $\hat{R}_{Z}(\theta)$, and hence one can obtain arbitrary rotations and nearest neighbour entangling gates. We show a particular example of implementing a nearest neighbour entangling gate with cluster state $|\text{OCS}_{2 \times 4}\rangle$ in Figure \ref{fig:cluster}. This completes the proof of Lemma \ref{lem:zrot}. 
\end{proof}
By using the lemmas proved above we can state the first universality theorem.
\begin{theorem}
\label{thm:univcs1}
The family of open-ended cluster states $|\text{inOCS}_{n \times m}\rangle$ is universal for quantum computation when used as a resource in MBQC with measurements limited to the $(X,Y)$-plane of the Bloch sphere.
\end{theorem}
\begin{proof}
Using Lemma \ref{lem:gateset} and Lemma \ref{lem:zrot} we conclude that any gate constructed using the gates from the universal set $\{ \hat{R}_{Z}(\theta), \hat{R}_{X}(\theta), \hat{R}_{ZX}(\theta)  \}$ can be implemented using $|\text{inOCS}\rangle$ as a resource state and $(X,Y)$-plane measurements. 
\end{proof}
  
\begin{corollary}
\label{thm:univcs}
The family of cluster states $|\text{CS}_{n \times m}\rangle$ is universal for quantum computation with the additional constraint that the measurement angles are chosen solely from the $(X,Y)$-plane.
\end{corollary}
\begin{proof}
In this section we will prove that universality of $|\text{inOCS}\rangle$ implies universality of $\CS$.  To show this, we note that a $|\text{inOCS}_{n \times m}\rangle$ and a $|\text{CS}_{n \times m}\rangle$ differ only by (possibly) the input state, and the $(n-1)$ $\text{Ctrl-Z}$ operators on the last column of the graph. 

Firstly, we see that universality of $|\text{inOCS}\rangle$ implies it can implement a circuit with $(n-1)$ $\text{Ctrl-Z}$ on a $n$-qubit state such that all the neighbouring qubits have $\text{Ctrl-Z}$ applied between them. Let us call such state $|{\text{CZ}_\text{inOCS}}\rangle$, which is given by $\prod_{i=1}^{n-1}\text{Ctrl-Z}_{i,i+1} |+\rangle^{\otimes n}$ and it is constructed using specific measurement angles on an open-ended cluster state $|\text{inOCS}_{n \times m}\rangle$. Then, any arbitrary unitary that can be constructed using the cluster state $|\text{inCS}_{n \times m}\rangle$ can be also constructed by concatenating the output of a cluster state $|\text{inOCS}_{n \times m}\rangle$ with the $|{\text{CZ}_\text{inOCS}}\rangle$ state.

Hence, using such construction, universality of $|\text{inOCS}\rangle$ (given by Theorem \ref{thm:univcs1}) implies universality of $|\text{inCS}\rangle$. However, one can think of any $\CS$ as the concatenation of two $|\text{inCS}\rangle$ states, the first with input state $|I\rangle = |+\rangle^{\otimes n}$ and the second with input state given by the output of the first. Then, universality of $|\text{inCS}\rangle$ immediately implies that the cluster state $\CS$ is universal with $(X,Y)$-plane measurements.
\end{proof}  

\emph{Acknowledgements} --- The authors acknowledge support from Singapore's Ministry of Education and National Research Foundation, and the Air Force Office of Scientific Research under AOARD grant FA2386-15-1-4082. This material is based on research funded in part by the Singapore National Research Foundation under NRF Award NRF-NRFF2013-01.

\bibliography{allrefs-Tom}
\bibliographystyle{unsrt}

\end{document}